\documentclass[letterpaper, 10 pt, conference]{ieeeconf}  
\IEEEoverridecommandlockouts                              
\usepackage[utf8]{inputenc}
\usepackage{amsmath}
\usepackage{amsthm}
\usepackage{mathrsfs}
\usepackage{amssymb}
\usepackage{color}
\usepackage{cite} 
\usepackage[font=footnotesize]{caption}
\usepackage{comment}
\usepackage{float}
\usepackage{graphicx}
\usepackage{BernsteinStyle}
\usepackage{float}
\usepackage{lscape}
\usepackage{bm}

\usepackage{tikz}
\usetikzlibrary{shapes,arrows,calc,positioning}
\tikzstyle{bigblock} = [draw, fill=blue!20, rectangle, 
    minimum height=6em, minimum width=8em]
\tikzstyle{medblock} = [draw, fill=blue!20, rectangle, 
    minimum height=4em, minimum width=4em]    
\tikzstyle{mux} = [draw, fill=black!20, rectangle, 
    minimum height=5em, minimum width=0.1em]    
\tikzstyle{smallblock} = [draw, fill=blue!20, rectangle, 
    minimum height=3em, minimum width=4em]
\tikzstyle{sum} = [draw, fill=blue!20, circle, node distance=1cm]
\tikzstyle{signal} = [coordinate]
\tikzstyle{pinstyle} = [pin edge={to-,thin,black}]
\tikzstyle{block} = [draw, fill=blue!20, rectangle, 
    minimum height=3em, minimum width=6em]
\tikzstyle{blockS} = [draw, fill=blue!20, rectangle, 
    minimum height=3em, minimum width=4em]    
\tikzstyle{input} = [coordinate]
\tikzstyle{output} = [coordinate]
\usetikzlibrary{matrix}

\newcounter{example}
\newenvironment{example}[1][]{\refstepcounter{example}\par\medskip
   \noindent \textbf{\indent Example~\theexample. #1} \rmfamily}{\medskip}

\newcommand{\Deltalog}{\Delta_{\rm log}}

\usepackage{array}

\newtheorem{lemma}{Lemma}

\title{A Numerical Investigation of the\\ Transient and Asymptotic Accuracy of Recursive Least Squares}

\title{A Numerical Investigation of the\\ Transient Performance and Convergence Rate of Recursive Least Squares}

\title{Regularization-Induced Bias and Consistency in Recursive Least Squares}

\author{\large Brian Lai, Syed Aseem Ul Islam,  and Dennis S. Bernstein 
\thanks{ Brian Lai, Syed Aseem Ul Islam, and Dennis S. Bernstein are with the Department of Aerospace Engineering, University of Michigan, Ann Arbor, MI, USA. {\tt\small \{brianlai, aseemisl, dsbaero\}@umich.edu}}
}

\begin{document}
\maketitle

\begin{abstract}
Within the context of recursive least squares (RLS) parameter estimation, the goal of the present paper is to study the effect of regularization-induced bias on the transient and asymptotic accuracy of the parameter estimates.
We consider this question in three stages.
First, we consider regression with random data, in which case persistency is guaranteed.
Next, we apply RLS to finite-impulse-response (FIR) system identification and, finally, to infinite-impulse-response (IIR) system identification.
For each case, we relate the condition number of the regressor matrix to the transient response and rate of convergence of the parameter estimates.
\end{abstract}

\section{Introduction}

In many parameter estimation problems, data becomes available in real time, and parameter estimates are desired at each successive step.
For least-squares estimation, recursive least square (RLS) is computationally efficient relative to batch least squares due to the need to invert a matrix whose dimension is lower than the parameter dimension \cite{albert1965,astrom,ljung:83,aseemrls}.
An additional advantage of RLS is the ability to include a forgetting factor, which weights more recent data more heavily than older data.
In effect, the forgetting factor facilitates learning in response to system changes by promoting forgetting.
Since system changes can occur sporadically and unexpectedly, forgetting must be adaptive and variable, see \cite{adamVRF2020} and the references therein.

Since batch least squares inverts a regressor matrix over a window of data, a full-rank regressor matrix is required before the first parameter estimate is available.
The invertibility of the regressor matrix depends on a persistency condition \cite{goel2020recursive}.
In the absence of persistency, however, a regularization term can be included, resulting in regularized batch least squares. 
Regularization is an essential technique in parameter estimation \cite{Hansen1993,Golub1999,Cucker2002,Lu2010}.
However, the regularization term perturbs the regressor, thus resulting in {\it regularization-induced bias}.
Along the same lines, RLS includes a regularization term in the cost function, which provides parameter estimates from the initial time.
As in the case of regularized batch least squares, the regularization term leads to bias in the parameter estimates of RLS.

For the case of constant regularization, the goal of the present paper is to study the effect of regularization-induced bias on the transient and asymptotic accuracy of the parameter estimates.
Note that this bias is due solely to the regularization and thus is not due to noise, which can also result in bias.
Other works have studied the bias induced by noise without regularization, \cite{adaptive_control_textbook} chapter 2, and the bias induced by noise and regularization together \cite{Banerjee2014}.
We consider this question in three stages.
First, we consider regression with random data, in which case persistency is guaranteed.
Next, we apply RLS to finite-impulse-response (FIR) system identification and, finally, to infinite-impulse-response (IIR) system identification.
For each case, we relate the condition number of the regressor matrix to the transient response and rate of convergence of the parameter estimates.

{\bf Notation and Terminology.}
Define $\BBN \isdef\{1,2,3,\dots\}$ and $\BBN_0 \isdef\{0\}\cup\BBN$.
The symbols $\mathbf{S}^n,$ $\mathbf{N}^n,$ and $\mathbf{P}^n$ denote the sets of real $n\times n$ symmetric, symmetric positive-semidefinite, and symmetric positive-definite matrices, respectively.
For $A\in \mathbf{N}^n$, $\lambda_i(A)$ denotes the $i$th largest eigenvalue of $A$, $\lambda_{\max}(A) \isdef \lambda_1(A)$, and $\lambda_{\min}(A) \isdef \lambda_n(A)$. 
Furthermore, the condition number, $\kappa (A)$, of $A \in \mathbf{P}^n$ is defined by
\begin{align}
    \kappa(A) \isdef \frac{\lambda_{\rm max}(A)}{\lambda_{\rm min}(A)}.
    \label{eq::Cond_numb_Defn}
\end{align}
%
If $A$ is positive-semidefinite but not positive-definite, then $\kappa(A)\isdef\infty.$  
$\|\cdot\|$ is the Euclidean norm.

\section{Batch Least Squares}

We consider the measurement process  
\begin{align}
    y_k = \phi_k\theta, \label{yphit}
\end{align}
where $k = 0,1,2,\ldots$ is the time step, 
$y_k\in\BBR^p$ is the measurement at step $k$,
the matrix $\phi_k\in\BBR^{p\times n}$ is the regressor at step $k$, and $\theta\in\BBR^n$ is a column vector of $n$ unknown parameters. 
The objective is to use  $y_k$ and $\phi_k$ to estimate the components of $\theta.$
Since, in practice, $y_k$ and $\phi_k$ are corrupted by noise, \eqref{yphit} does not hold exactly, and we thus consider least squares estimates of $\theta$.
The batch approach to this problem is to collect a large amount of data and then apply least squares optimization to compute an estimate of $\theta.$
In particular, using data from the step window $i=0,1,\ldots,k,$ it follows from \eqref{yphit} that
\begin{align}
    Y_k = \Phi_k \theta, \label{YPhiT}
\end{align}
where
\begin{align}
    Y_k \isdef \left[ \begin{array}{c}
         y_0 \\
         \vdots \\
         y_k
    \end{array}\right]
    ,\quad \Phi_k \isdef \left[ \begin{array}{c}
         \phi_0 \\
         \vdots \\
         \phi_k
    \end{array}\right].
\end{align}
Note that \eqref{YPhiT} has the form $Ax=b,$ where $A$ denotes $\Phi_k,$ $x$ denotes $\theta,$ and $b$ denotes $Y_k.$

In the presence of noise corrupting the data $Y$ and $\Phi,$ \eqref{YPhiT} may not have a solution.
In this case, it is useful to replace \eqref{YPhiT} by the least-squares optimization of the cost
\begin{align}
     J_{R,k}(\hat\theta) &\isdef \sum_{i=0}^{k}( y_i -  \phi_i \hat\theta )^\rmT ( y_i - \phi_i \hat\theta) + (\hat\theta- \theta_0)^\rmT  R (\hat\theta- \theta_0 )\nn\\
    &=(Y_k-\Phi_k\hat\theta)^\rmT (Y_k-\Phi_k \hat\theta)+ (\hat\theta- \theta_0)^\rmT  R (\hat\theta- \theta_0 ), 
    \label{RBLS}
\end{align}
where $R \in \mathbf{P}^n$ and $\theta_0 \in \BBR^n$ is an initial estimate of $\theta.$ 
The regularization term $(\hat\theta- \theta_0)^\rmT  R (\hat\theta- \theta_0 )$ weighs the distance from the current estimate to the initial estimate and ensures that $ J_{R,k}$ has a unique global minimizer.
%
%
In particular, the {\it batch least-squares} (BLS) minimizer of \eqref{RBLS} is given by
\begin{align}
    \theta_{{\rm BLS},R,k+1} &\defeq \underset{ \hat\theta \in \BBR^n  }{\operatorname{argmin}} \ J_{R,k}(\hat\theta) \\
    &= (\Phi_k^\rmT \Phi_k^{} +R)^{-1}(\Phi_k^\rmT Y_k+R\theta_0),
    \label{LSsoln}
\end{align}
where the inverse $(\Phi_k^\rmT \Phi_k^{}+R)^{-1}$ exists due to the positive-definite regularization $R.$

Note that $(\Phi_k^\rmT \Phi_k^{}+R)^{-1}$ requires the computation of an $n\times n$ inverse,
and thus the computational requirement of the inverse is of order $n^3.$
Note also that the memory needed to store $\Phi_k$ grows with $k.$
Furthermore, if $\Phi_k$ has full column rank, then the regularization is not needed, and thus $R$ can be set to zero.  
In this case, \eqref{LSsoln} becomes
\begin{align}
\theta_{\rm BLS,0,k+1} = (\Phi_k^\rmT \Phi_k^{} )^{-1}\Phi_k^\rmT Y_k,\label{LSsoln2}
\end{align}
which is the unique solution of \eqref{YPhiT}.
%
\section{Recursive Least Squares}
Recursive least squares (RLS) provides a recursive update of the minimizer of \eqref{RBLS} as measurements become available.
Although RLS can be stated with a forgetting factor $\lambda$, the following result provides a statement of RLS with $\lambda = 1.$ 
\begin{theo}\label{thm1}
For all $k\in \BBN_0$, let $\phi_k\in \BBR^{p\times n}$ and $y_k \in\BBR^p.$
Furthermore, let $\theta_0 \in \BBR^{n}$, let $P_0 \in \BBR^{n\times n}$ be positive definite.
Furthermore, for all $k\in \BBN_0$, denote the minimizer of the function $J_{P_0^{-1},k}(\hat\theta)$ by 
\begin{align}
    \theta_{k+1} \isdef \underset{ \hat\theta \in \BBR^n  }{\operatorname{argmin}} \ J_{P_0^{-1},k}(\hat\theta).\label{thetakplus1}
\end{align}
Then, for all $k\in \BBN_0$, $\theta_{k+1}$ is given by
\begin{align}
    P_{k+1} &=P_{k}   - P_{k}\phi_k^\rmT (  I +  \phi_kP_{k}\phi_k^\rmT )^{-1} \phi_k P_{k},\label{eq:P_update_noInverse}\\
    \theta_{k+1} &= \theta_{k}  + P_{k+1}\phi_k^\rmT (   y_k - \phi_k \theta_{k}  ). \label{eq:theta_update_noInverse}
\end{align}
\end{theo}
Note, that, for all $k\ge 1,$
\begin{align}
    \theta_{k} = \theta_{{\rm BLS},P_0^{-1},k}.
\end{align}
\textbf{This implies any results for RLS and BLS are equivalent, where $\bm{R = P_0^{-1}}$.} However, the computational requirements of RLS are primarily determined by the inverse in \eqref{eq:P_update_noInverse}, which is of size $p\times p.$   When $p \ll n,$ the computational burden of this inverse is much less demanding than the $n\times n$ inverse required by BLS.
In addition, the storage requirements of RLS are of order $n^2,$ which does not grow with $k.$
Consequently, the computational and memory requirements of RLS are significantly less than those of BLS.
For future analysis, we use the notation   
\begin{align}
    \theta \isdef \left[ \begin{array}{c}
         \theta_{(1)} \\
         \vdots \\
         \theta_{(n)}
    \end{array}\right]
    ,\quad \theta_k \isdef \left[ \begin{array}{c}
         \theta_{k,(1)} \\
         \vdots \\
         \theta_{k,(n)}
    \end{array}\right].
\end{align}
\section{Regularization-Induced Bias} 
Note that $\theta_{{\rm BLS},R,k+1}$ given by \eqref{LSsoln} does not equal $\theta_{{\rm BLS},0,k+1}$ given by \eqref{LSsoln2}, and thus, the regularization $R$ induces a bias in the estimate of $\theta.$
This bias decreases as more data are available, as demonstrated by the following example.
\begin{example}
\label{ex1}
For all $k \in \BBN_0$ let the components of $\theta \in \mathbb{R}^{10}$ and $\phi_k \in \mathbb{R}^{1 \times 10}$ be sampled from the uniform distribution on $[-1,1]$, and thus $y_k$ is scalar. 
Let $R = P_0^{-1} = rI$, where $r$ is a positive number and let $\theta_{0,(i)} = 0$ for $i = 1,\hdots, n$. 
For each value of $r$, 100 independent simulations are run, and, at each step $k,$ the estimation error $\|\theta_k - \theta\|$ is averaged over the 100  simulations.
Figure \ref{Reg-Induced-Bias} shows the averaged estimation error with $k\in[0,10^4].$
Note that $\|\theta_k - \theta\|$ sharply decreases around $k= 10,$ and that the decrease is larger for smaller values of $r.$
In addition, $\|\theta_k - \theta\|$  converges with the same linear log-log slope for all values of $r$.
\hfill{\large$\diamond$}
\begin{figure}[ht]
    \begin{center}
     \includegraphics[trim = 0mm 0mm 15mm 0mm, clip, width=0.45\textwidth]{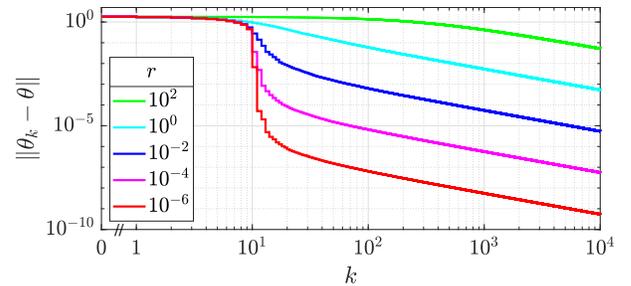}
    \caption{\footnotesize Estimation error. The regularization is given by $R = rI.$   For each value of $r,$ 100 independent simulations are run where, for all $k\in \BBN_0$, each component of $\theta \in \mathbb{R}^{10}$ and $\phi_k \in \mathbb{R}^{1 \times 10}$ sampled from the uniform distribution on $[-1,1]$. The norm of the estimation error is averaged over the 100 trials for $k\in[0,10^4]$.}
    \label{Reg-Induced-Bias}
    \end{center}
\end{figure}
\end{example}
\section{Consistency of the Regularized Solution}
The identities 
\begin{align}
    P_{k+1}^{-1} &= P_k^{-1} + \phi_k^\rmT \phi_k, \label{eq:Pinveqn} \\
    P_{k}^{-1} &= P_0^{-1} + \sum_{i=0}^{k-1}\phi_i^\rmT \phi_i\label{eq:Pinvtot}, \\ 
    \theta_k - \theta &= P_k P_0^{-1}(\theta_0 - \theta),\label{eq:xtildetot} \\
    \theta_k &= (I-P_k P_0^{-1})\theta +  P_k P_0^{-1} \theta_0\label{eq:xktot}
\end{align}
are useful. Proofs of these identities can be found in \cite{goel2020recursive}. The next results follows. 
\begin{prop}
Consider the notation and assumptions of Theorem 1.  Then the following statements hold:
\begin{enumerate}
    \item For all $k\in \BBN_0$, $P_{k+1} \le P_k$.
    \label{RLS_prop_1}
    \item $P_\infty \isdef \lim_{k\to\infty}P_k$ exists, and $P_\infty\in\BFN^n$.  
    \label{RLS_prop_2}
    \item $\theta_{\infty} \isdef \lim_{k\to\infty} \theta_k$ exists, and  $\theta_{\infty}=(I-P_\infty P_0^{-1})\theta +  P_\infty P_0^{-1} \theta_0$.
    \label{RLS_prop_3}
    \item  $\theta_{\infty}=\theta$ if and only if $ P_{\infty} P_0^{-1}\theta=  P_{\infty} P_0^{-1}\theta_0.$
    \label{RLS_prop_4}
    \item  If $P_\infty=0$, then  $\theta_{\infty}  =\theta$.
    \label{RLS_prop_5}
    \item If, for all $\theta_0\in\BBR^n,$ $\theta_{\infty}  =\theta,$ then $P_\infty=0$. 
    \label{RLS_prop_6}
\end{enumerate}
\end{prop} 
\noindent {\it Proof.} 
\begin{enumerate}
    \item Note \eqref{eq:Pinveqn} and that, for all $k\in \BBN_0$, $\phi_k^\rmT\phi_k \in \BFN^n$ . 
    \item Note that, for all $k\in \BBN_0$, $P_k \in \BFN^n$ and $P_{k+1} \le P_k$.
    \item The result follows directly from \eqref{eq:xktot}.
    \item Substitute $\theta_\infty = \theta$ into the result of \ref{RLS_prop_3}.
    \item Substitute $P_\infty=0$ into the second equation of \ref{RLS_prop_4}.
    \item Let $\theta_0 \in \BBR^n$ and $P_\infty \ne 0$. Then, there exists $\theta \in \BBR^n$ such that $P_{\infty} P_0^{-1}\theta \ne  P_{\infty} P_0^{-1}\theta_0$. So, by \ref{RLS_prop_4}, $\theta_\infty \ne \theta$.
\end{enumerate}
\begin{flushright}
\mbox{$\square$}
\end{flushright}
\subsection{Persistently Exciting Regressors}
\begin{defin}\label{def:reg_WP}  $(\phi_k)_{k=0}^\infty\subset \BBR^{p\times n}$ is {\it persistently exciting} (PE) if there exist $\alpha > 0$, $\beta > 0$, and $N \in \BBN_0$ such that, for all $j \in\BBN_0$,
\begin{align}
    \alpha I_n \le \sum_{i=0}^{N} \phi_{i+j}^\rmT\phi_{i+j} \le \beta I_n.\label{eq:persistsum}
\end{align}
\end{defin}
\begin{prop}
$(\phi_k)_{k=0}^\infty\subset \BBR^{p\times n}$ is PE if and only if
\begin{align}
    C \defeq \lim_{k \rightarrow \infty} \frac{1}{k} \Phi_k^\rmT \Phi_k 
    \label{eq::PE_matrix_C}
\end{align}
exists and is positive definite. 
\end{prop}
\begin{proof}
See page 64 of \cite{adaptive_control_textbook}.
\end{proof}
%
%
%
%
%
%
\begin{lemma}
Let $(A_k)_{k=0}^{\infty} \subset \BBR^{n \times n}$ and assume that,  for all $k\in \BBN_0$, $A_k$ is nonsingular, and $A\isdef \lim_{k \rightarrow \infty} A_k$ exists and is nonsingular. Then, $\lim_{k \rightarrow \infty} A_k^{-1}$ exists and $\lim_{k \rightarrow \infty} A_k^{-1} = A^{-1}$.
\label{matrix_inverse_lemma}
\end{lemma}
\begin{proof}
This result follows from the continuity of the matrix inverse on the set of nonsingular matrices.
\end{proof}

Under persistent excitation, the following result describes the asymptotic rate of convergence of the regularized RLS solution.



\begin{theo}
\label{theo::PE_implies_1/k}
Let $\theta,\theta_0 \in \BBR^n$, for all $k\in\BBN_0,$ let $y_k$ be given by \eqref{yphit}, and let $P_k$ and $\theta_k$ be given by \eqref{eq:P_update_noInverse} and \eqref{eq:theta_update_noInverse}, respectively. 
Assume that $(\phi_k)_{k=0}^\infty$ is PE,  and define $C$ by \eqref{eq::PE_matrix_C}.  Then,
$\lim_{k \rightarrow \infty}k(\theta_k - \theta) = C^{-1} R (\theta_0 - \theta)$.
\end{theo}

\noindent {\it Proof.}
Since
\begin{align*}
    \theta_k &= (\Phi_{k-1}^\rmT \Phi_{k-1} + R)^{-1}(\Phi_{k-1}^\rmT Y_{k-1} + R\theta_0) \\
    &= (\Phi_{k-1}^\rmT \Phi_{k-1} + R)^{-1}(\Phi_{k-1}^\rmT \Phi_{k-1}\theta + R\theta - R\theta + R\theta_0) \\
    &= (\Phi_{k-1}^\rmT \Phi_{k-1} + R)^{-1}((\Phi_{k-1}^\rmT \Phi_{k-1} + R)\theta + R(\theta_0 - \theta)) \\
    &= \theta + (\Phi_{k-1}^\rmT \Phi_{k-1} + R)^{-1} R (\theta_0 - \theta),    
\end{align*}    
it follows that 
\begin{align}
    \theta_k - \theta = (\Phi_{k-1}^\rmT \Phi_{k-1} + R)^{-1} R (\theta_0 - \theta).
    \label{eq::theta-theta_k}
\end{align}
Hence, by Lemma \ref{matrix_inverse_lemma},
\begin{align}
    \lim_{k \rightarrow \infty} k(\theta_k - \theta) &= \lim_{k \rightarrow \infty} (\frac{1}{k}\Phi_{k-1}^\rmT \Phi_{k-1} + \frac{1}{k}R)^{-1} R (\theta_0 - \theta)\nn\\
    & = C^{-1}R (\theta_0 - \theta).\tag*{\mbox{$\square$}}
\end{align}
%


Assume that $\theta_0 \ne \theta$ and that the assumptions of Theorem \ref{theo::PE_implies_1/k} hold. Then, $v\isdef C^{-1} R (\theta_0 - \theta )\ne0.$       
Therefore, for all $i = 1,\hdots,n$, $(\theta_{k,(i)} - \theta_{(i)}) = O(1/k)$ as $k\to\infty.$
If, in addition, $v_{(i)}=0,$ then $(\theta_{k,(i)} - \theta_{(i)}) = o(1/k)$ as $k\to\infty.$

\subsection{The Condition Number}
\begin{lemma}
If $A \in \mathbf{P}^n$, then, for all nonzero $\alpha\in\BBR,$ $\kappa(A) = \kappa(\alpha A)$.
\label{cond_numb_scaled_matrix}
\end{lemma}
%
%

\begin{lemma}
Let $(A_k)_{k=0}^{\infty} \subset \mathbf{P}^n$ and assume that  $A\isdef\lim_{k\to\infty}A_k$  exists.     Then $\kappa(A_k) \rightarrow \kappa(A)$ as $k \rightarrow \infty$. 
\label{lemma:cond_num_of_sequence}
\end{lemma}
%
%

\begin{prop}
Under the assumptions and notation of Theorem \ref{theo::PE_implies_1/k}, $\lim_{k \rightarrow \infty} \kappa(\Phi_k^\rmT\Phi_k + R) = \kappa(C)\ge1$.
\label{prop:cond_num}
\end{prop}
%




\noindent {\it Proof.}       
%
Note that
\begin{align}
    \lim_{k \rightarrow \infty} \frac{1}{k}(\Phi_k^\rmT\Phi_k + R) = \lim_{k \rightarrow \infty} \frac{1}{k}\Phi_k^\rmT\Phi_k.
    \label{cond_num_proof_eq1}
\end{align}
Next, by Lemma \ref{cond_numb_scaled_matrix}, for all $k\in \BBN_0$,   $\kappa(\Phi_k^\rmT\Phi_k + R) = \kappa(\frac{1}{k}(\Phi_k^\rmT\Phi_k + R))$. Hence, 
\begin{align}
    \lim_{k \rightarrow \infty} \kappa(\Phi_k^\rmT\Phi_k + R) = \lim_{k \rightarrow \infty} \kappa\left(\frac{1}{k}(\Phi_k^\rmT\Phi_k + R)\right).
    \label{cond_num_proof_eq2}
\end{align}
Then, by Lemma \ref{lemma:cond_num_of_sequence}, 
\begin{align}
    \lim_{k \rightarrow \infty} \kappa(\Phi_k^\rmT\Phi_k + R) = \kappa\left(\lim_{k \rightarrow \infty} \frac{1}{k}(\Phi_k^\rmT\Phi_k + R)\right).
    \label{cond_num_proof_eq3}
\end{align}
Finally,   \eqref{cond_num_proof_eq1} implies
\begin{align}
    \lim_{k \rightarrow \infty} \kappa(\Phi_k^\rmT\Phi_k + R) = \kappa\left(\lim_{k \rightarrow \infty} \frac{1}{k}\Phi_k^\rmT\Phi_k\right) = \kappa(C) \ge 1. \tag*{\mbox{$\square$}} 
\end{align}
%


%
The condition number of the regressor matrix can be used as a metric to assess the convergence of the regularized RLS solution under persistent excitation, particularly when $\theta$ is unknown. Note that the convergence of the condition number of $\Phi_k^\rmT\Phi_k + R$ is  necessary but not sufficient  for the convergence of $\theta_k \rightarrow \theta$ as $k \rightarrow \infty$ under persistent excitation.

\section{Analysis of Consistency} 
This section presents numerical examples with diagnostics in order to investigate the properties the estimation error shown in Figure 1, namely, the sudden decrease in the estimation error and its asymptotic log-log slope.
To do this, we apply RLS to \eqref{yphit}, where the components of $\phi_k$ are randomly generated and $y_k$ is given by \eqref{yphit}.

\begin{example}\label{eg2}
Let
\begin{align}
    \theta = \begin{bmatrix} 1 & 1 & 1 & 1 \end{bmatrix}^{\rmT},
\end{align}
and $\theta_{0} = 0$. For all $k\in \BBN_0$, assume $\phi_k \in \BBR^{1 \times 4}$ are  i.i.d. and are randomly generated by 
\begin{align}
    \phi_k^\rmT \sim \mathcal{N}\big(0,\operatorname{diag}(\sigma_1^2,\sigma_2^2,\sigma_3^2,\sigma_4^2)\big),
\end{align}
where $\sigma_1^2 = 0.1$, $\sigma_2^2 = 1$, $\sigma_3^2 = 10$, and $\sigma_4^2 = 100$. Let $R = rI,$ where $r = 10^{-5}$.  Then,  
\begin{align}
    C = \lim_{k \rightarrow \infty}\frac{1}{k}\Phi_k^\rmT \Phi_k = \operatorname{diag}(\sigma_1^2,\sigma_2^2,\sigma_3^2,\sigma_4^2).
\end{align}
Hence, Theorem \ref{theo::PE_implies_1/k} implies that
\begin{align}
    \lim_{k \rightarrow \infty} k(\theta_k - \theta) &= 
    -r\begin{bmatrix}
    \sigma_1^{-2} & \sigma_2^{-2} & \sigma_3^{-2} & \sigma_4^{-2}
    \end{bmatrix}^\rmT.
    \label{ex2_asymptotic_slope}
\end{align}
Next, by Proposition \ref{prop:cond_num},
\begin{align}
    \lim_{k \rightarrow \infty} \kappa(\Phi_k^\rmT\Phi_k + rI) = \kappa(C) = \frac{\sigma_4^2}{\sigma_1^2} = 10^3.
    \label{ex2_cond_numb}
\end{align}
RLS is run for 10 independent simulations and, at each step $k,$ the estimation error $\|\theta_k - \theta\|$ is averaged over the 10 simulations. Figure \ref{ex2figure_error} shows the error $|\theta_{(m),k} - \theta_{(m)}|$ for $m = 1,2,3,4$ and the condition number $\kappa{(\Phi_k^{\rmT}\Phi_k + rI)}$, as well as the asymptotic behaviors given by \eqref{ex2_asymptotic_slope} and \eqref{ex2_cond_numb}.
%
%
\begin{figure}[ht]
    \begin{center}
     \includegraphics[trim = 0mm 0mm 0mm 0mm, clip, width=0.48\textwidth]{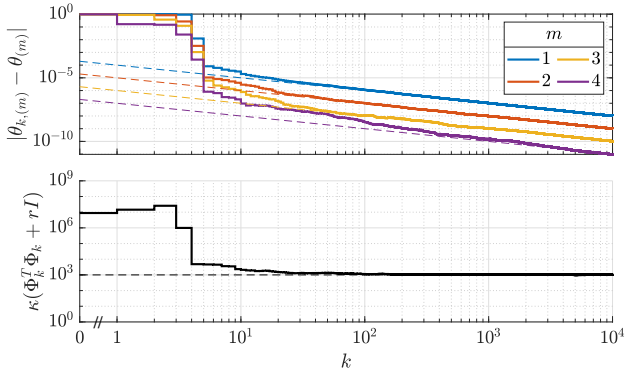}
    \caption{\footnotesize Example \ref{eg2}: Estimation error (top) as well as condition number of   $\Phi_k^\rmT\Phi_k + rI$ (bottom) for $k \in [0,10^4]$, averaged over 10 independent trials.
    The regressor $\phi_k$ is sampled $\phi_k^\rmT \sim \mathcal{N}\big(0,\operatorname{diag}(\sigma_1^2,\sigma_2^2,\sigma_3^2,\sigma_4^2)\big)$, where $\sigma_1^2 = 0.1$, $\sigma_2^2 = 1$, $\sigma_3^2 = 10$, and $\sigma_4^2 = 100$. The true parameters are $\theta = [ 1 \ 1 \ 1 \ 1 ]^\rmT$, and the regularization is $R = 10^{-5}I$. 
    The dashed lines show the asymptotic behaviors of the estimation error and condition number given by \eqref{ex2_asymptotic_slope} and \eqref{ex2_cond_numb} respectively.}
    \label{ex2figure_error}
    \end{center}
\end{figure}
\\
For convenience,  we define the {\it log slope} by
\begin{align}
    \Deltalog(f_k) \defeq \frac{\log(f_k) - \log(f_{k-1})}{\log(k) - \log(k-1)},
    \label{log_slope_defn}
\end{align}
which is independent of the logarithm base. 
%


\begin{prop}
Let $\theta,\theta_k,\theta_{k-1} \in \mathbb{R}^n$, let $p  \ge 1$, and assume there exists $c \in \mathbb{R}$ such that, for all $m = 1,\ldots n$,
\begin{align}
    \Deltalog(|\theta_{k,(m)} - \theta_{(m)}|) = c.
    \label{lemma1IF}
\end{align}
%
%
%
Then
\begin{align}
    \Deltalog(\|\theta_k - \theta\|_p) = c.
    \label{lemma1THEN}
\end{align}
\label{lemmaSameSlope}
\end{prop}













\noindent {\it Proof.} 
\eqref{log_slope_defn} and \eqref{lemma1IF} imply that, for all $m = 1,\ldots,n$,  
\begin{align}
    \frac{|\theta_{k,(m)} - \theta_{(m)}|}{|\theta_{k-1,(m)} - \theta_{(m)}|} = q \isdef
    \left(\frac{k}{k-1}\right)^c.
\end{align}
%
%
Next, note that,
\begin{align*}
    \frac{\|\theta_{k} - \theta\|_p}{\|\theta_{k-1} - \theta\|_p} 
    = \left(\frac{\sum_{m=1}^n |\theta_{k,(m)}-\theta_{(m)}|^p}{\sum_{m=1}^n |\theta_{k-1,(m)}-\theta_{(m)}|^p}\right)^{1/p}\\
    = \left(\frac{\sum_{m=1}^n |\theta_{k-1,(m)} - \theta_{(m)}|^p q^p}{\sum_{m=1}^n |\theta_{k-1,(m)}-\theta_{(m)}|^p}\right)^{1/p} = (q^p)^{1/p} = q.
\end{align*}
Thus, \eqref{log_slope_defn} yields
\begin{align}
    \Deltalog(\|\theta_k - \theta\|_p) 
    = \frac{ \log \left(  \frac{\|\theta_{k} - \theta\|_p}{\|\theta_{k-1} - \theta\|_p} \right) }{\log\left( \frac{k}{k-1} \right)}
    =\frac{c\log\left(\frac{k}{k-1}\right)}{\log\left( \frac{k}{k-1} \right)} = c.
    \tag*{\mbox{$\square$}}
\end{align}
Note that \eqref{ex2_asymptotic_slope} implies that, for all $m = 1,\hdots,n$, $\Deltalog(|\theta_{k,(m)} - \theta_{(m)}|) \rightarrow -1$ as $k \rightarrow \infty$. Hence, by Proposition \ref{lemmaSameSlope}, $\Deltalog\|\theta_k - \theta\| \rightarrow -1$ as $k \rightarrow \infty$. Therefore, Figure \ref{ex2figure_slopes} and all subsequent plots of log slope show $\Deltalog\|\theta_k - \theta\|$. 
Note that, in Figure \ref{ex2figure_slopes}, $\Deltalog\|\theta_k - \theta\|$ is also averaged over the 10 independent trials.
Moreover, since we are concerned with the asymptotic log slope of the estimation error, we plot $\Deltalog\|\theta_k - \theta\|$ for $ 9000 \le k \le 10000.$
As shown in Figure \ref{ex2figure_slopes}, the 100-step moving average of $\Deltalog\|\theta_k - \theta\|$ is approximately equal to $-1$.
Furthermore, note that $\kappa(\Phi_{k}^{\rmT}\Phi_{k} + rI)$ in Figure \ref{ex2figure_error} approaches $10^3$. 
The change in condition number $\kappa{(\Phi_{k}^{\rmT}\Phi_{k} + rI)}$ for $k \in [9000,10000]$ is shown in Figure \ref{ex2figure_slopes}, where
\begin{align}
\Delta(\kappa{(\Phi_{k}^{\rmT}\Phi_{k} + rI)}) \defeq \kappa{(\Phi_{k}^{\rmT}\Phi_{k} + rI)} - \kappa{(\Phi_{k-1}^{\rmT}\Phi_{k-1} + rI)}.    
\end{align}
$\hfill\mbox{$\huge\diamond$}$
\begin{figure}[ht]
    \begin{center}
     \includegraphics[trim = 0mm 0mm 0mm 0mm, clip, width=0.48\textwidth]{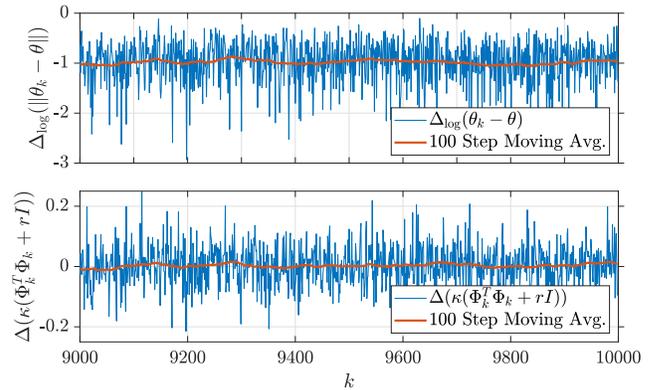}
    \caption{\footnotesize Example \ref{eg2}: Log slope of error (top) and $\Delta(\kappa{(\Phi_{k}^{\rmT}\Phi_{k} + rI)})$ (bottom) for $k \in [9000,10000]$, averaged over 10 independent trials. The regressor $\phi_k$ is sampled $\phi_k^\rmT \sim \mathcal{N}\big(0,\operatorname{diag}(\sigma_1^2,\sigma_2^2,\sigma_3^2,\sigma_4^2)\big)$, where $\sigma_1^2 = 0.1$, $\sigma_2^2 = 1$, $\sigma_3^2 = 10$, and $\sigma_4^2 = 100$. The 100-step moving averages show that $\Deltalog(\|\theta_k - \theta\|) \rightarrow -1$ and  $\kappa{(\Phi_{k}^{\rmT}\Phi_{k} + rI)} \rightarrow 10^3$ as $k \rightarrow \infty$.}
    \label{ex2figure_slopes}
    \end{center}
\end{figure}

\end{example}

\section{Application to FIR System Identification} 

Let a finite impulse response (FIR) be given by
\begin{align}
    y_k = \sum_{i=1}^w{G_i u_{k-i}}, \label{FIR}
\end{align}
for $k = 0,1,2,\hdots$, where $w$ is the  window length, $u_k \in \mathbb{R}^q$ is the  input, $y_k \in \mathbb{R}^p$ output, and $G_i \in \mathbb{R}^{p \times q}$ are the model coefficients. 
Note that \eqref{FIR} can be expressed as \eqref{yphit}, where
\begin{align}
    \phi_k &\defeq \begin{bmatrix} u_{k-1}^\rmT \ \hdots \ u_{k-w}^\rmT \end{bmatrix} \otimes I_p \in \mathbb{R}^{p \times n}\label{phi_k FIR},\\
    \theta &\defeq {\rm vec}{\begin{bmatrix} G_{1} \ \hdots \ G_{w} \end{bmatrix}} \in \mathbb{R}^{n}. \label{theta FIR}
\end{align}
To estimate $\theta$,  define
\begin{align}
    \theta_k &\defeq {\rm vec}{\begin{bmatrix} G_{1,k} \ \hdots \ G_{w,k} \end{bmatrix}} \in \mathbb{R}^{n}, \label{theta_k FIR}
\end{align}
where $G_{i,k} \in \mathbb{R}^{p \times m}$ are the estimated model coefficients and $n = wpq$. 
Note that $u_k \defeq 0$ for all $k < 0$.
\begin{example}\label{eg3}
\textit{FIR model}
Consider the FIR model
\begin{align} \label{FIR model}
    y_k = -1.5u_{k-1} + 0.9u_{k-2} + 0.15u_{k-3} - 0.15u_{k-4},
\end{align}
for all $k\in \BBN_0$, where $\theta$ is defined by \eqref{theta FIR}. 
For all $k\in \BBN_0$, the inputs $u_k \in \mathbb{R}$ are i.i.d and are sampled from the uniform distribution on $[-1,1]$.
RLS is applied to estimate $\theta$ with regularization $R = rI$ where $r = 10^{-5}$ and $\theta_0 = 0$. 
Figure \ref{eg3_error} shows the error $|\theta_{(m),k} - \theta_{(m)}|$ for $m = 1,2,3,4$ as well as the condition number $\kappa{(\Phi_k^{\rmT}\Phi_k + rI)}$.

Note that the error is   $O(10^{-1})$ at step $k = m$, but decreases to   $O(10^{-5})$ at step $k = m+1$, for all $m = 1,2,3,4$, as seen in Figure \ref{eg3_error}. 
Additionally, Figure \ref{eg3_error} shows that the condition number of $\Phi_k^\rmT\Phi_k + rI$ drops below $10^2$ at $k=4$, the step at which $\Phi_k^\rmT\Phi_k$ attains full rank.   
Similar to the results of Example \ref{eg2}, Figure \ref{eg3_slopes} shows that $\Deltalog(\|\theta_k - \theta\|)$  approaches $-1$ and $\kappa(\Phi_k^\rmT\Phi_k + rI)\approx 1$ asymptotically. 
\begin{figure}[ht]
    \begin{center}
    \vspace{1.5em}
     \includegraphics[trim = 0mm 0mm 0mm 0mm, clip,  width=0.48\textwidth]{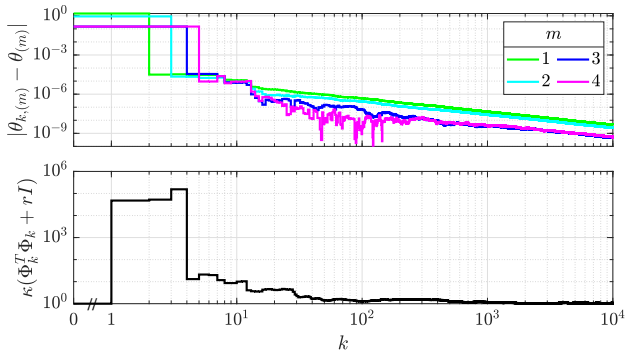}
    \vspace{-.2in}
    \caption{\footnotesize Example \ref{eg3}: Estimation error (top) and condition number of  $\Phi_k^\rmT\Phi_k = rI$ (bottom) for FIR system identification with steps $k \in [0,10^4]$. For all $k\in \BBN_0$, the input $u_k$ is sampled from the uniform distribution on $[-1,1]$. The true parameters are $\theta = [ -1.5 \ 0.9 \ 0.15 \ -0.15 ]^\rmT$, and the regularization is $R = 10^{-5}I$.} 
    \label{eg3_error}
    \end{center}
\end{figure}
\begin{figure}[ht]
    \begin{center}
    \vspace{.08in}
     \includegraphics[trim = 0mm 0mm 0mm 0mm, clip,  width=0.48\textwidth]{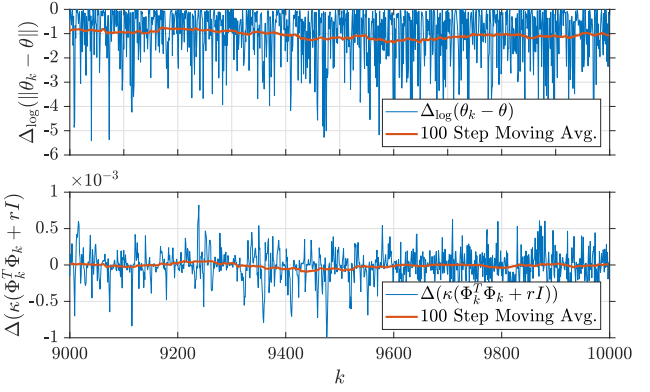}
    \vspace{-.2in}
    \caption{\footnotesize Example \ref{eg3}: Log slope of error (top) and $\Delta(\kappa(\Phi_k^\rmT\Phi_k + rI))$ (bottom) for FIR system identification  with $k \in [9000,10000].$  For all $k\in \BBN_0$, the input $u_k$ is sampled from the uniform distribution on $[-1,1]$. The 100-step moving averages show that $\Deltalog(\|\theta_k - \theta\|)$ approaches $-1$ and  $\kappa{(\Phi_{k}^{\rmT}\Phi_{k} + rI)} \approx 1$ as $k \rightarrow \infty.$} 
    \label{eg3_slopes}
    \end{center}
\end{figure}
$\hfill\mbox{$\huge\diamond$}$
\end{example}
\section{Application to IIR System Identification} 
Let an infinite impulse response (IIR) be given by
\begin{align}
    y_k = -\sum_{i=1}^w{F_i y_{k-i}} + \sum_{i=1}^w{G_i u_{k-i}}, 
    \label{IIR}
\end{align}
for $k = 0,1,2,\hdots$, where $w$ is the model window length, $u_k \in \mathbb{R}^q$ is the input, $y_k \in \mathbb{R}^p$ is the output, and $F_i \in \mathbb{R}^{p \times p}$ and $G_i \in \mathbb{R}^{p \times q}$ are the model coefficients.
Note that \eqref{IIR} can be expressed as \eqref{yphit}, where
\begin{align}
    \theta &\defeq {\rm vec}{\begin{bmatrix} F_{1} \ \hdots \ F_{w} \ G_{1} \ \hdots \ G_{w} \end{bmatrix}} \in \mathbb{R}^{n}, \label{theta IIR} \\
    \phi_k &\defeq \begin{bmatrix} -y_{k-1}^\rmT  \hdots  -y_{k-w}^\rmT \ u_{k-1}^\rmT  \hdots  u_{k-w}^\rmT \end{bmatrix} \otimes I_p \in \mathbb{R}^{p \times n}. \label{phi_k IIR}
\end{align}
To estimate $\theta$,  define
\begin{align}
    \theta_k &\defeq {\rm vec}{\begin{bmatrix} F_{1,k} \ \hdots \ F_{w,k} \ G_{1,k} \ \hdots \ G_{w,k} \end{bmatrix}} \in \mathbb{R}^{n}, \label{theta_k IIR} \\
\end{align}
where $F_{i,k} \in \mathbb{R}^{p \times p}$ and $G_{i,k} \in \mathbb{R}^{p \times q}$ are the estimated model coefficients and $n = wp(p+q)$. 
Note that $y_k \defeq 0$ and $u_k \defeq 0$ for all $k < 0$. 
\begin{example} \label{eg4}
\textit{IIR model}
Consider the IIR model
\begin{align}
    y_k = 1.5y_{k-1} - 0.9y_{k-2} + 0.15u_{k-1} - 0.15u_{k-2},
\end{align}
for all $k\in \BBN_0$, where $\theta$ is defined by \eqref{theta IIR}. 
For all $k\in \BBN_0$, the inputs $u_k \in \mathbb{R}$ are i.i.d. and are sampled from the uniform distribution on $[-1,1]$.
RLS is applied to estimate $\theta$ with regularization $R = rI$ where $r = 10^{-5}$ and $\theta_0 = 0$.
Figure \ref{eg4_error} shows the error $|\theta_{(m),k} - \theta_{(m)}|$ for $m = 1,2,3,4$ as well as the condition number $\kappa{(\Phi_k^{\rmT}\Phi_k + rI)}$. 
\begin{figure}[ht]
    \begin{center}
    \vspace{.08in}
     \includegraphics[trim = 0mm 0mm 0mm 0mm, clip,  width=0.48\textwidth]{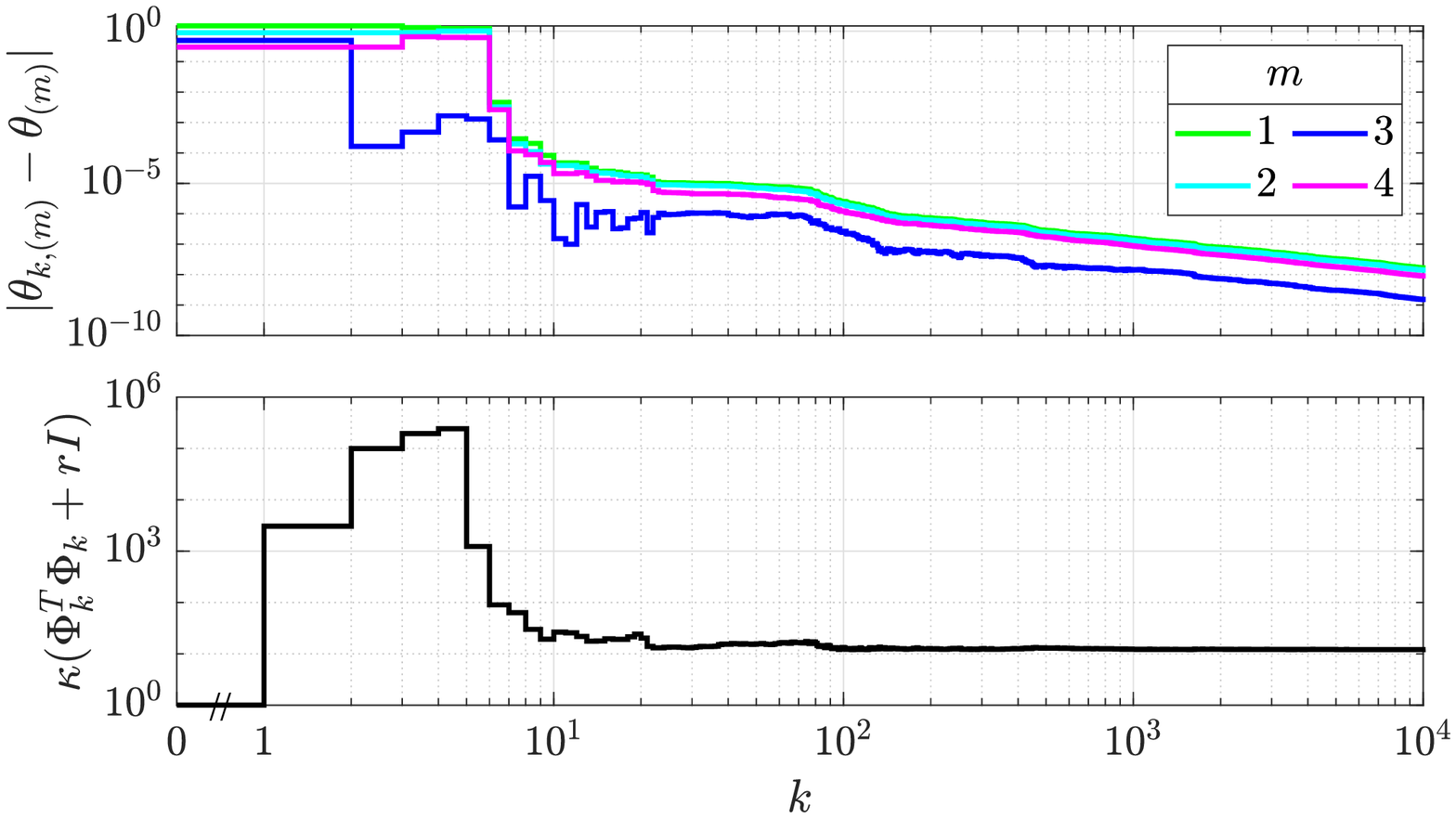}
    \vspace{-.2in}
    \caption{\footnotesize Example \ref{eg4}: Estimation error (top) and condition number of  $\Phi_k^\rmT\Phi_k = rI$ (bottom) for IIR system identification with steps $k \in [0,10^4]$. For all $k\in \BBN_0$, the input $u_k$ is sampled from the uniform distribution on $[-1,1]$. The true parameters are $\theta = [ -1.5 \ 0.9 \ 0.15 \ -\mspace{-2mu}0.15 ]^\rmT$, and the regularization is $R = 10^{-5}I$.} 
    \label{eg4_error}
    \end{center}
\end{figure}
\begin{figure}[ht]
    \begin{center}
     \includegraphics[trim = 0mm 0mm 0mm 0mm, clip,  width=0.48\textwidth]{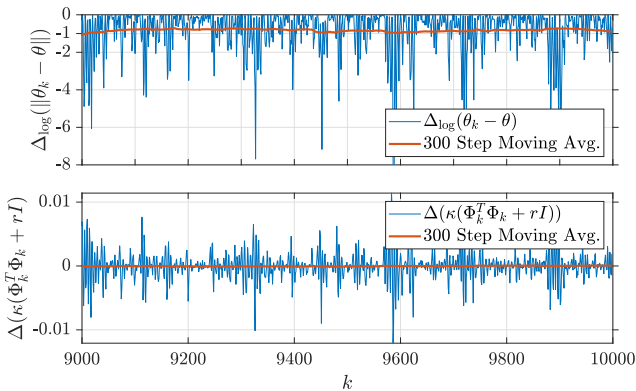}
    \vspace{-.2in}
    \caption{\footnotesize Example \ref{eg4}: Log slope of error (top) and $\Delta(\kappa(\Phi_k^\rmT\Phi_k + rI))$ (bottom) for IIR system identification  with $k \in [9000,10000]$. For all $k\in \BBN_0$, the input $u_k$ is sampled from the uniform distribution on $[-1,1]$. The 100-step moving averages show that $\Deltalog(\|\theta_k - \theta\|)$ approaches $-1$ and  $\kappa{(\Phi_{k}^{\rmT}\Phi_{k} + rI)} \approx 10$ as $k \rightarrow \infty$.} 
    \label{eg4_slopes}
    \end{center}
\end{figure}

Note that $|\theta_{k,(m)} - \theta_{(m)}|$ decreases approximately four orders of magnitude from step $k=1$ to $k=2$ for $m=3$. 
Additionally, $|\theta_{k,(m)} - \theta_{(m)}|$ further decreases approximately three orders of magnitude from steps $k=5$ to $k=6$ for $m = 1,2,4$.

Note that, for $k > 5$, the matrix $\Phi_k^\rmT\Phi_k + rI $ is ill-conditioned, with $\kappa{(\Phi_k^{\rmT}\Phi_k + rI)} \approx 10$ for all $k > 100$.
Yet, this ill-conditioned $(\Phi_k^{\rmT}\Phi_k + rI)$ for $k > 5$ does not affect the asymptotic behavior of $|\theta_{k,(m)} - \theta_{(m)}|$ for all values of $m$.
In fact, Figure \ref{eg4_slopes} shows that $\Deltalog(\|\theta_k - \theta\|)$ approaches $-1$ and  $\kappa(\Phi_k^{\rmT}\Phi_k + rI) \approx 10$ asymptotically.

$\hfill\mbox{$\huge\diamond$}$
\end{example}

\section{Conclusions and Future Work}
In optimization-based parameter estimation, regularization compensates for the lack of persistency, especially during startup.
This paper examined regularization-induced bias, which refers to bias in the parameter estimates due to regularization. 
This paper showed that, under persistency, the parameter estimates improve at precisely the step where the regressor becomes square and nonsingular.
It was also shown that the regularization-induced bias asymptotically decreases with a log slope of approximately $-1.$ 
Both of these effects were connected to the condition number of the regressor.

For system identification, RLS was applied to FIR and IIR systems.
For FIR systems, it was shown that the components of the estimation error decrease sequentially rather than simultaneously.
However, for IIR systems, the components of the estimation error decrease less predictably.
Future research will relate these trends to the condition number of the regressor.

Finally, one potential approach to overcoming regularization-induced bias is to use variable regularization \cite{AliHoaggMossbergBernstein}.
By varying the regularization based on the condition number of the regressor, it may be possible to reduce the regularization-induced bias.
The challenge is to vary the regularization in a computationally efficient manner without impacting the numerical stability of RLS.

\bibliographystyle{IEEEtran}
\bibliography{main}

\end{document}